\documentclass[11pt, letterpaper]{article}
\usepackage[utf8]{inputenc}
\usepackage{fullpage,times}

\usepackage{graphicx}
\usepackage[intlimits]{amsmath}
\usepackage{amssymb}
\usepackage{amsthm}
\usepackage{exscale}
\usepackage{thmtools}
\usepackage{thm-restate}
\usepackage[hypertexnames=false]{hyperref}
\usepackage{times}
\usepackage{xspace}
\usepackage{footnote}
\usepackage[noend]{algpseudocode} 
\usepackage{subfigure}
\usepackage[english]{babel}
\usepackage[utf8]{inputenc}
\usepackage{amsfonts,amssymb,latexsym,mathtools,tikz-cd,float,bm,mathdots,amsmath,todonotes}
\usepackage{algorithm}
\usepackage[noend]{algpseudocode}
\newcommand{\commentsymbol}{//}
\algrenewcommand\algorithmiccomment[1]{\hfill \commentsymbol{} #1}
\makeatletter
\newcommand{\LineComment}[2][\algorithmicindent]{\Statex \hspace{#1}\commentsymbol{} #2}
\makeatother

\newtheorem{theorem}{Theorem}[section]

\newtheorem{claim}[theorem]{Claim}

\theoremstyle{definition}
\newtheorem{definition}[theorem]{Definition}

\usepackage{enumerate}
\usepackage{enumitem}
\usepackage{tcolorbox}

\DeclareMathOperator{\poly}{poly}

\definecolor{mygreen}{RGB}{20,140,80}
\definecolor{linkcolor}{RGB}{0,0,230}
\definecolor{mylightgray}{RGB}{230,230,230}
\definecolor{verylightgray}{RGB}{245,245,245}

\newcounter{myalgctr}

\newtcolorbox{OuterBox}[1][]{%
    breakable,
    enhanced,
    frame hidden,
    interior hidden,
    left=-5pt,
    right=-5pt,
    top=-5pt,
    float=p,
    boxsep=0pt,
    arc=0pt
#1}%

\newtcolorbox{InnerBox}[1][]{%
    enforce breakable,
    enhanced,
    colback=gray,
    colframe=white,
#1}%

\newenvironment{tbox}{
\vspace{0.2cm}
\begin{tcolorbox}[width=\columnwidth,
                  enhanced,
                  boxsep=2pt,
                  left=1pt,
                  right=1pt,
                  top=4pt,
                  boxrule=1pt,
                  arc=0pt,
                  colback=white,
                  colframe=black,
	              breakable
                  ]
}{
\end{tcolorbox}
}

\newcommand{\tboxhrule}[0]{\vspace{0.1cm} {\color{black} \hrule} \vspace{0.2cm}}

\newenvironment{titledtbox}[1]{\begin{tbox}#1 \tboxhrule}{\end{tbox}}

\setlist[1]{itemsep=-.1pt}

\usepackage[capitalize,nameinlink]{cleveref}
\crefname{theorem}{Theorem}{Theorems}
\Crefname{lemma}{Lemma}{Lemmas}
\Crefname{claim}{Claim}{Claims}
\Crefname{observation}{Observation}{Observations}
\Crefname{algorithm}{Algorithm}{Algorithms}
\Crefname{myalgctr}{Algorithm}{Algorithms}
\Crefname{challenge}{Challenge}{Challenges}

\title{Algorithms for the Minimum Dominating Set Problem in Bounded Arboricity Graphs: Simpler, Faster, and Combinatorial}

\author{Adir Morgan\footnote{Email: adirmorgan@gmail.com} \\Tel Aviv University \and Shay Solomon\footnote{Partially supported by the Israel Science Foundation grant No.1991/19. Email: solo.shay@gmail.com} \\ Tel Aviv University \and Nicole Wein\footnote{Supported by NSF Grant CCF-1514339. Email: nwein@mit.edu} \\ MIT}
\date{}

\begin{document}
\maketitle

\begin{abstract}
We revisit the minimum dominating set problem on graphs with arboricity bounded by $\alpha$.
In the (standard) centralized setting, Bansal and Umboh~\cite{bansal2017tight} gave an $O(\alpha)$-approximation LP rounding algorithm, which also translates into a near-linear time algorithm using general-purpose approximation results for explicit mixed packing and covering or pure covering LPs
\cite{koufogiannakis2014nearly,young2014nearly,allen2019nearly,quanrud2020nearly}.
Moreover, \cite{bansal2017tight} showed that it is NP-hard to achieve an asymptotic improvement for the approximation factor. 
On the other hand, the previous two {\em non}-LP-based algorithms, by Lenzen and Wattenhofer~\cite{LW10}, and Jones et al.~\cite{jones2013parameterized}, achieve an approximation factor of $O(\alpha^2)$ in {\em linear} time.
 
There is a similar situation in the distributed setting: While there is an $O(\log^2 n)$-round LP-based  $O(\alpha)$-approximation algorithm implied in~\cite{kuhn2006price}, the best non-LP-based algorithm by Lenzen and Wattenhofer~\cite{LW10} is an implementation of their centralized algorithm, providing an $O(\alpha^2)$-approximation within $O(\log n)$ rounds.

We address the questions of whether one can achieve an
$O(\alpha)$-approximation algorithm that is
{\em elementary}, i.e., not based on any LP-based methods, either in the centralized setting or in the distributed setting.
We resolve both questions in the affirmative, and en route achieve algorithms that are faster than the state-of-the-art LP-based algorithms. More specifically, our contribution is two-fold:
\begin{enumerate}
\item In the centralized setting, we provide a surprisingly simple combinatorial algorithm that is asymptotically optimal in terms of both approximation factor and running time: an $O(\alpha)$-approximation in {\em linear} time. 
The previous state-of-the-art  $O(\alpha)$-approximation algorithms are (1) LP-based,  (2) more complicated, and (3) have super-linear running time. 
\item Based on our centralized algorithm, we design a distributed combinatorial $O(\alpha)$-approximation algorithm in the $\mathsf{CONGEST}$ model that runs in $O(\alpha\log n )$ rounds with high probability. 
Not only does this result provide the first nontrivial {\em non}-LP-based distributed $o(\alpha^2)$-approximation algorithm for this problem, it also outperforms the best LP-based distributed algorithm for a wide range of parameters.
\end{enumerate}

\end{abstract}

\setcounter{page}0
\thispagestyle{empty}
\clearpage

\section{Introduction}

\subsection{Background}
The minimum dominating set (MDS) problem is a classic combinatorial optimization problem. Given a graph $G$ we want to find a minimum cardinality set $D$ of vertices, such that every vertex of the graph is either in $D$ or has a neighbor in $D$.  Besides its theoretical implications, solving this basic problem efficiently has many practical applications in domains ranging from wireless networks to text summarizing (see, e.g., \cite{wan2002distributed, nacher2016minimum, SL10}).
The MDS problem was one of the first problems recognized as NP-complete \cite{garey1979guide}. It was also one of the first problems for which an approximation algorithm was analyzed: a simple greedy algorithm achieves a $\ln n$-approximation in general graphs~\cite{johnson1974approximation}. This approximation factor is optimal up to lower order terms unless $\mathsf{P}=\mathsf{NP}$~\cite{dinur2014analytical}. 

\paragraph{Distributed MDS in general graphs} The first efficient distributed approximation algorithm for MDS was given by Jia, Rajaraman, and Suel~\cite{jia2002efficient}, who gave a randomized $O(\log \Delta)$-approximation in $O(\log^2 n)$ rounds in the $\mathsf{CONGEST}$ model. This was improved by Kuhn, Moscibroda and Wattenhofer \cite{kuhn2006price}, who gave a randomized  $(1+\varepsilon)(1+\ln(\Delta+1))$-approximation in $O(\log^2 \Delta/\varepsilon^4)$ rounds in the $\mathsf{CONGEST}$ model and in $O(\log n /\varepsilon^2)$ rounds in the $\mathsf{LOCAL}$ model. Ghaffari, Kuhn, and Maus~\cite{ghaffari2017complexity} showed that by allowing exponential-time local computation, one can get a randomized $(1+o(1))$-approximation in a polylogarithmic number of rounds in the $\mathsf{LOCAL}$ model. This result was derandomized by the network decomposition result of Rozho{\v{n}} and Ghaffari \cite{rozhovn2020polylogarithmic}. From the lower bounds side, Kuhn, Moscibroda, and Wattenhofer \cite{kuhn2016local} showed that getting a polylogarithmic approximation ratio requires $\Omega\big(\sqrt{\frac{\log n}{\log\log n}}\big)$ and $\Omega\big(\frac{\log \Delta}{\log\log\Delta}\big)$ rounds in the $\mathsf{LOCAL}$ model.


For \emph{deterministic} distributed algorithms, improving over previous work, Deurer, Kuhn, and Maus \cite{deurer2019deterministic} recently gave two algorithms in the $\mathsf{CONGEST}$ model with approximation factor $(1+\varepsilon)\ln(\Delta + 1)$ for $\varepsilon>1/\text{polylog}\Delta$, running in  $2^{O(\sqrt{\log n \log\log n})}$ and $O((\Delta+\log^*n)\text{polylog}\Delta)$ rounds, respectively; the running time of the former $\mathsf{CONGEST}$ algorithm \cite{ghaffari2018derandomizing}, achieving approximation factor $O(\log^2 n)$,
is dominated by the time needed for deterministically  computing  a  network  decomposition  in  the  $\mathsf{CONGEST}$ model, which, due to \cite{ghaffari2021improved}, is thus reduced to $O(\poly \log n)$.




\paragraph{Graphs of bounded arboricity} The MDS problem has been studied on a variety of restricted classes of graphs, such as graphs with bounded degree (e.g., \cite{chlebik2008approximation}), planar and bounded genus graphs (e.g., \cite{baker1994approximation, czygrinow2008fast,amiri2019distributed}), and graphs of bounded arboricity --- which is the focus of this paper. The class of bounded \emph{arboricity} graphs is a wide family of \emph{uniformly sparse} graphs, defined as follows:

\begin{definition} Graph $G$ has {\em arboricity} bounded by $\alpha$ if for every $S\subseteq V$, it holds that $\frac{m_s}{n_s-1}\leq \alpha$, where $m_s$ and $n_s$ are the number of edges and vertices in the subgraph induced by $S$, respectively.
\end{definition}

The class of bounded arboricity graphs contains the other graph classes mentioned above as well as bounded treewidth graphs, and in general all graphs excluding a fixed minor. Moreover, many natural and real world graphs, such as the world wide web graph, social networks and transaction networks, are believed to have bounded arboricity.
Consequently, this class of graphs has been subject to extensive research, which led to many algorithms for bounded arboricity graphs in both the (classic) centralized setting (e.g. \cite{eppstein1994arboricity, goel2006bounded, chiba1985arboricity}) and in the distributed setting (e.g. \cite{czygrinow2009fast, barenboim2010sublogarithmic,GS17,SuV20});
there are also many algorithms in other settings, such as dynamic graph algorithms, sublinear algorithms and streaming algorithms (see \cite{BF99,HTZ14,PS16,parter2016local,OSSW18,solomon2018improved,kaplan2021dynamic,ELR18,ERR19,ERS20,BPS20,MV18,BS20,bera2020graph}, and the references therein).

In distributed settings, one cannot always assume that all processors know the arboricity of the graph, so it is important to devise robust algorithms, which can perform correctly also when the arboricity is unknown to the processors (see e.g.~\cite{barenboim2010sublogarithmic, LW10}).

\subsection{Approximating MDS on graphs of arboricity $\alpha$}
\paragraph*{Centralized setting}
In the centralized setting, there are two non-LP-based
algorithms for MDS for graphs of arboricity (at most) $\alpha$ (for brevity, in what follows we may write graphs of ``arboricity $\alpha$'' instead of arboricity {\em at most} $\alpha$). One is by Lenzen and Wattenhofer~\cite{LW10}, the other is by Jones,  Lokshtanov, Ramanujan, Saurabh, and Such{\`y}~\cite{jones2013parameterized}, and both achieve an $O(\alpha^2)$-approximation in deterministic linear time\footnote{Note that the theorem statement of~\cite{LW10} has a typo suggesting that the approximation factor is $O(\alpha)$.}. There is also a very simple LP rounding algorithm by Bansal and Umboh that gives a $3\alpha$-approximation~\cite{bansal2017tight}. This algorithm is very simple, after the LP has been solved. To solve the LP, there are near-linear time general-purpose approximation algorithms for explicit mixed packing and covering or pure covering LPs
\cite{koufogiannakis2014nearly,young2014nearly,allen2019nearly,quanrud2020nearly}. Combining such an algorithm with \cite{bansal2017tight} yields an $O(\alpha)$-approximation for MDS, either deterministically within $O(m \log n)$ time \cite{young2014nearly} or randomly (with high probability) within $O(n \log n + m)$ time \cite{koufogiannakis2014nearly}.
The latter  bound is super-linear in the entire (non-degenerate) regime of arboricity $\alpha = o(\log n)$;
the regime $\alpha = \Omega(\log n)$ is considered degenerate, since in that case one can use the greedy linear-time $\ln n$-approximation algorithm.
Bansal and Umboh~\cite{bansal2017tight} also proved that achieving asymptotically better approximation is NP-hard.\footnote{More specifically, achieving an $(\alpha-1-\varepsilon)$-approximation is NP-hard for any $\varepsilon>0$ and any fixed $\alpha$; achieving an $(\lfloor \alpha / 2 \rfloor - \varepsilon)$-approximation is NP-hard for any $\varepsilon > 0$ and any $\alpha = 1,\ldots,\log^\delta n$, for some constant $\delta$ \cite{bansal2017tight,DGKR05}.
These hardness of approximation results are achieved by applying a reduction by \cite{bansal2017tight}
from the $k$-hypergraph vertex cover ($k$-HVC)
problem (where we need to find a minimum vertex cover of a $k$-uniform hypergraph) to the MDS problem in arboricity-$k$ graphs, in conjunction with NP-hardness results by \cite{DGKR05} for the $k$-HVC problem.}


\paragraph*{Distributed setting} In the distributed setting, there are two 
non-LP-based 
algorithms for MDS for graphs of arboricity $\alpha$, both by Lenzen and Wattenhofer~\cite{LW10}. The first is a randomized $O(\alpha^2)$-approximation algorithm in the $\mathsf{CONGEST}$ model that runs in $O(\log n)$ rounds with high probability. This algorithm was made deterministic by Amiri \cite{amiri2021deterministic}, and uses an LP-based subroutine of Even, Ghaffari, and Medina \cite{even2018distributed}. The second algorithm of Lenzen and Wattenhofer is a deterministic $O(\alpha\log \Delta)$-approximation algorithm in the $\mathsf{CONGEST}$ model that runs in $O(\log\Delta)$ rounds, where $\Delta$ is the maximum degree.

Regarding LP-based algorithms, Kuhn, Moscibroda, and Wattenhofer~\cite{kuhn2006price} developed a general-purpose method for solving LPs of a particular structure in the distributed setting. It seems that by applying their method (specifically, Corollary 4.1 of~\cite{kuhn2006price}) to the LP approximation result of Bansal and Umboh in bounded arboricity graphs~\cite{bansal2017tight}, one can get a deterministic $O(\alpha)$-approximation algorithm for MDS in the $\mathsf{CONGEST}$ model that runs in $O(\log^2 \Delta)$ rounds, but such a result has not been explicitly claimed in the literature.

\paragraph*{A natural question}
The aforementioned results demonstrate a significant gap for MDS algorithms in bounded arboricity graphs when comparing LP-based methods to elementary {\em combinatorial} approaches. It is natural to ask whether this gap can be bridged. 

\begin{itemize}
\item 
In the centralized setting, is there any efficient non-LP-based $O(\alpha)$-approximation algorithm for MDS (even one that is slower than the aforementioned $O(m \log n)$ time deterministic and $O(n\log n + m)$ time randomized LP-based algorithms)? Further, can one achieve an $O(\alpha)$-approximation in {\em linear time} using {\em any} (even LP-based) algorithm?

\item
In the distributed setting, is there any efficient non-LP-based distributed $O(\alpha)$-approximation algorithm for MDS?
Further, can one achieve an $O(\alpha)$-approximation in the $\mathsf{CONGEST}$ model within $o(\log^2 \Delta)$ rounds using {\em any} (even LP-based) algorithm?
\end{itemize}

We note the caveat that there is no clear-cut distinction between combinatorial and non-combinatorial algorithms, but we operate under the premise that an algorithm is combinatorial if all its intermediate computations have a natural combinatorial interpretation in terms of the original problem. While all algorithms presented in this paper are certainly combinatorial under this premise, it is far less clear whether prior work is. In particular, the previous state-of-the-art LP-based approaches are based on
general-purpose primal/dual methods; when restricted to the MDS problem, it is possible that these methods could reduce, after proper adaptations, into simpler combinatorial algorithms. Nonetheless, even if possible, it is unlikely that the resulting algorithm would be as simple and elementary as ours. 
In the distributed setting, Kuhn and Wattenhofer \cite{kuhn2005constant} give an LP-based algorithm specifically for MDS that is simpler than the subsequent general-purpose LP-based algorithm of Kuhn, Moscibroda, and Wattenhofer \cite{kuhn2006price}; however,  \cite{kuhn2005constant} is inferior to \cite{kuhn2006price} in both approximation ratio and running time.

\subsection{Our Contributions}

We answer all parts of the above question in the affirmative. In particular, we give algorithms that achieve the asymptotically optimal approximation factor of $O(\alpha)$,
and are not only simple and elementary, but also run faster than all known algorithms, including LP-based algorithms.\footnote{$O(\alpha)$ is the asymptotically optimal approximation factor for polynomial time algorithms in the centralized setting and also in distributed settings where processors are assumed to have polynomially-bounded processing power.}

\subsubsection*{Centralized Setting}

Our core contribution is an asymptotically optimal algorithm in the centralized setting.

\begin{restatable}{theorem}{thmstatic}
\label{thm:static}
For graphs of arboricity $\alpha$, there is an $O(m)$ time $O(\alpha)$-approximation algorithm for MDS.
\end{restatable}

We note that our algorithm works even when $\alpha$ is not known a priori, since there is a linear time 2-approximation algorithm for computing the arboricity of a graph~\cite{arikati1997efficient}.

Our algorithm is asymptotically optimal in both running time and approximation factor: it runs in linear time, and asymptotically improving the approximation factor it gets is proved to be NP-hard~\cite{bansal2017tight}. (The constant in the approximation ratio is not tight; our algorithm gives an $8\alpha$-approxmation.)
While the quantitative improvement in running time over prior work is admittedly minor (a logarithmic factor over the deterministic algorithm, and $\log n / \alpha$ over the randomized algorithm), still getting a truly linear time algorithm is qualitatively very different than an almost-linear time. Indeed, the study of linear time algorithms has received much attention over the years, even when it comes to shaving factors that grow as slowly as inverse-Ackermann type functions.
This line of work includes celebrated breakthroughs in computer science: For example, for the Union-Find data structure, efforts to achieve a linear time algorithm led to a lower bound showing that inverse-Ackermann function dependence is necessary~\cite{fredman1989cell}, matching the upper bound~\cite{tarjan1975efficiency}, which is a cornerstone result in the field. Another example is MST, where the inverse-Ackermann function was shaved from the upper bound of~\cite{chazelle2000minimum} to achieve a linear time algorithm either using randomization~\cite{karger1995randomized} or when the edge weights are integers represented in binary~\cite{fredman1990trans}, but it remains a major open problem whether or not there exists a linear time deterministic comparison-based MST algorithm.

\subsubsection*{Distributed Setting}

We demonstrate the applicability of our centralized algorithm, by using its core ideas to develop a distributed algorithm.

\begin{restatable}{theorem}{thmimproved}
\label{thm:improved}
For graphs of arboricity $\alpha$, there is a randomized distributed algorithm in the $\mathsf{CONGEST}$ model that gives an $O(\alpha)$-approximation for MDS and runs in $O(\alpha\log n)$ rounds. The bound on the number of rounds holds with high probability (and in expectation). The algorithm works even when either $\alpha$ or $n$ is unknown to each processor.
\end{restatable}


For the ``interesting'' parameter regime where $\Delta$ is polynomial in $n$, and $\alpha=o(\log n)$, the number of rounds in our algorithm beats the prior work obtained by combining~\cite{kuhn2006price} and~\cite{bansal2017tight} which appears to run in $O(\log^2 \Delta)$ rounds; as noted already, such an algorithm has not been claimed explicitly before. We note the caveat that our algorithm is randomized while their algorithm appears to be deterministic.

In the process of obtaining our distributed algorithm, we also obtain a \emph{deterministic} algorithm in the $\mathsf{LOCAL}$ model (with polynomial message sizes) in a polylogarithmic number of rounds, via reduction to the maximal independent set (MIS) problem:

\begin{restatable}{theorem}{thmlocal}
\label{thm:local}
Suppose there is a deterministic (resp., randomized) distributed algorithm in the $\mathsf{LOCAL}$ model for computing an MIS on a general graph in $R(n)$ rounds. Then, for graphs of arboricity $\alpha$, there is a deterministic (resp., randomized) distributed algorithm in the $\mathsf{LOCAL}$ model that gives an $O(\alpha)$-approximation for MDS in $O(R(n)\cdot\alpha^2\log n)$ rounds. The algorithm works even when either $\alpha$ or $n$ is unknown to each processor.
\end{restatable}

While \cref{thm:local} is the first deterministic
non-LP-based 
algorithm to achieve an $O(\alpha)$-approximation, we note that the LP-based approach obtained by combining ~\cite{kuhn2006price} and~\cite{bansal2017tight} appears to achieve fewer rounds and work in the $\mathsf{CONGEST}$ model.
\cref{thm:local} is not our main result and is used as a stepping stone towards our $O(\alpha\log n)$ round algorithm in the $\mathsf{CONGEST}$ model.

We finally note that unlike in the centralized setting, handling unknown $\alpha$ in the distributed setting it is not trivial and requires special treatment; in \cref{sec:unknown} we demonstrate that all of our distributed algorithms can cope with unknown $\alpha$ without increasing the approximation factor and running time beyond constant factors.

\paragraph*{Wider applicability}
We have demonstrated the applicability of our centralized algorithm to the distributed setting. We anticipate that the core idea behind our centralized algorithm could be applied more broadly, to other settings that involve locality. Perhaps the prime example in this context is the standard (centralized) setting of dynamic graph algorithms, where the graph undergoes a sequence of edge updates (a single edge update per step), and the algorithm should maintain the graph structure of interest ($O(\alpha)$-approximate MDS in our case) with a small {\em update time} --- preferably  $\poly \log(n)$ and ideally $O(1)$. 


\subsection{Technical overview}
\paragraph*{Centralized algorithm} As a starting point, we consider the algorithm of Jones,  Lokshtanov, Ramanujan, Saurabh, and Such{\`y}~\cite{jones2013parameterized}, which achieves an $O(\alpha^2)$-approximation in linear time. Their algorithm is as follows. They iteratively build a dominating set $D$ and maintain a partition of the remaining vertices into the dominated vertices $B$ (the vertices that have a neighbor in $D$), and the undominated vertices $W$. This partition of the vertices, as well as further partitioning described later, is shown in \cref{fig:dom}. The basic property of arboricity $\alpha$ graphs used by their algorithm is that every subgraph contains a vertex of degree $O(\alpha)$. They begin by choosing a vertex $v$ with degree $O(\alpha)$ and adding $v$  along with $v$'s entire neighborhood $N(v)$ to $D$. The intuition behind this is that at least one vertex in $\{v\}\cup N(v)$ must be in $OPT$ (an optimal dominating set), since $OPT$ must dominate $v$. Hence, they add at least one vertex in $OPT$ and use that to pay for adding $O(\alpha)$ vertices not in $OPT$. We say that a vertex $w$ \emph{witnesses} $v$ and the vertices in $N(v)$ that are added to $D$, if $w\in OPT\cap (\{v\}\cup N(v))$. Now, the goal of the algorithm is to iteratively choose vertices $v$ to add to $D$ along with $O(\alpha)$ many of $v$'s neighbors so that each vertex in $OPT$ witnesses $O(\alpha)$ vertices $v$ along with $O(\alpha)$ neighbors for each such vertex $v$. That is, each vertex in $OPT$ witnesses $O(\alpha^2)$ vertices in $D$, which yields an $O(\alpha^2)$-approximation.

To choose which vertices $v$ and which $O(\alpha)$ of $v$'s neighbors to add to $D$, they partition the set $B$ into two subsets $B_{low}$ and $B_{high}$, which are the sets of vertices in $B$ with low and high degree to $W$, respectively, where the degree threshold is $\delta\alpha$ for some constant $\delta$. We also define $W_{low}\subseteq W$ (differently from the notation of \cite{jones2013parameterized}) as the subset of vertices with degree at most $\delta\alpha$ in the subgraph induced by $W\cup B_{high}$. They add a vertex $w\in W_{low}$ to $D$ along with $w$'s $O(\alpha)$ neighbors that are in $W\cup B_{high}$. In the interest of brevity, we will not motivate why this scheme achieves the desired outcome that each vertex in $OPT$ witnesses $O(\alpha^2)$ vertices in $D$. 

The key innovation in our algorithm that allows us to reduce the approximation factor from $O(\alpha^2)$ to $O(\alpha)$ is a simple but powerful idea. After choosing a vertex $w$ to add to $D$, we do not {\em immediately} add $O(\alpha)$ of $w$'s neighbors to $D$. Instead $w$ casts a ``vote'' for these $O(\alpha)$ neighbors, and only once a vertex gets $\delta\alpha$ many votes is it added to $D$. With this modification, we can argue that each vertex in $OPT$ still witnesses $O(\alpha)$ such vertices $w$ as in the previous approach, but the catch here is that each such vertex $w$ contributes only $O(1)$ neighbors to $D$ on average, so each vertex in $OPT$ only witnesses a \emph{total} of $O(\alpha)$ vertices in $D$, rather than $O(\alpha^2)$.  Moreover, it is straightforward to implement this algorithm in linear time.

\begin{figure}[h]
    \centering
    \includegraphics{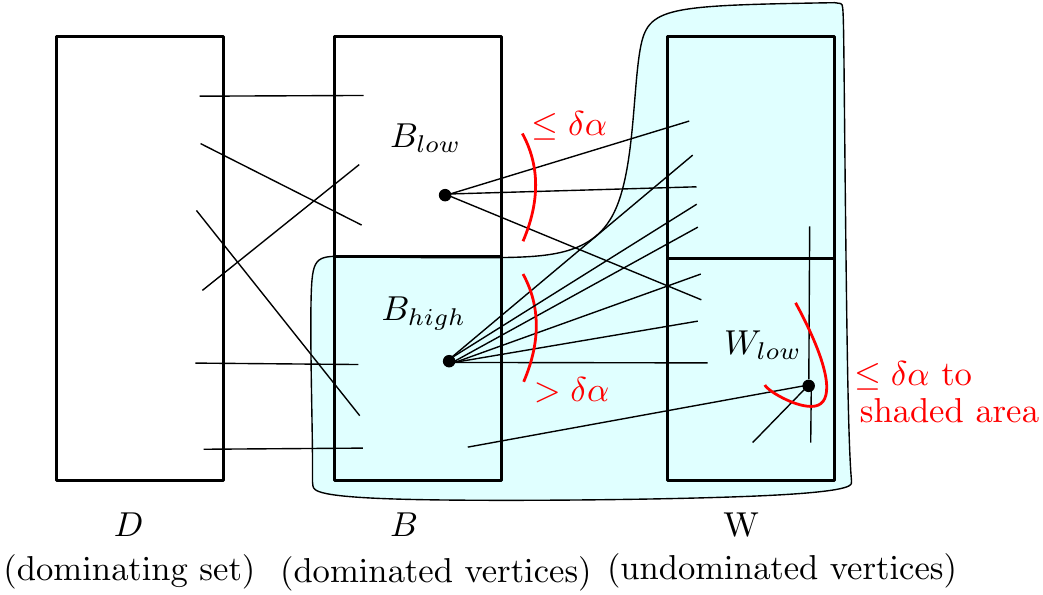}
    \caption{The partition of vertices}
    \label{fig:dom}
\end{figure}

\paragraph{Distributed algorithms using MIS} This section concerns the proof of~\cref{thm:local}: our reduction from MDS to MIS in the $\mathsf{LOCAL}$ model. This section also concerns a modification of this reduction that gives an $O(\alpha^2\log^2 n)$ round algorithm in the $\mathsf{CONGEST}$ model. We use this algorithm as a stepping stone towards obtaining our main distributed algorithm (\cref{thm:improved}) which runs in $O(\alpha\log n)$ rounds in the $\mathsf{CONGEST}$ model. 

We adapt our centralized algorithm to the distributed setting as follows. Recall that in our centralized algorithm, we repeatedly choose a vertex $w\in W_{low}$, add $w$ to $D$, and cast a vote for each vertex in $N(w)\cap (W\cup B_{high})$. For our distributed algorithms, we would like to choose \emph{many} such vertices $w$ and process them in {\em parallel}. In fact, a constant fraction of the vertices in $W\cup B_{high}$ could be chosen as our vertex $w$ since a constant fraction of vertices in a graph of arboricity $\alpha$ have degree $O(\alpha)$. However, we cannot simply process all of these vertices in parallel. In particular, if a vertex $v$ has many neighbors being processed in parallel, $v$ might accumulate many votes during a single round. This would invalidate the analysis of the algorithm, which relies on the fact that once a vertex $v$ receives $\delta\alpha$ votes, $v$ enters $D$. 

To overcome this issue, we compute an MIS with respect to a 2-hop graph built from a subgraph of ``candidate" vertices, and only process the vertices in this MIS in parallel. This MIS has two useful properties: 1. Its maximality implies that in any 2-hop neighborhood of a candidate vertex there is a vertex in the MIS; this helps to bound the number of rounds, and 2. Its independence implies that every vertex has at most one neighbor in the MIS, which ensures that any vertex can only receive one vote per round. To conclude, this approach gives a reduction from distributed MDS to distributed MIS in the $\mathsf{LOCAL}$ model. This approach can be made to work in the $\mathsf{CONGEST}$ model by replacing the black-box MIS algorithm with a 2-hop version of Luby's algorithm. This approach of running the 2-hop version of Luby's algorithm was also used in \cite{LW10} for their distributed $(\alpha^2)$-approximation for MDS.

\paragraph{Faster randomized distributed algorithm}  In the $\mathsf{CONGEST}$ model, our distributed algorithm using MIS runs in $O(\alpha^2 \log^2 n)$ rounds with high probability. We devise a new, more nuanced algorithm that decreases the number of rounds to $O(\alpha \log n)$ with high probability. Our new algorithm is based on our previous algorithm, but with two key modifications, which save factors of $\log n$ and $\alpha$, respectively.

Our first key modification, which shaves a $\log n$ factor from the number of rounds, is that we do not run an MIS algorithm as a black box. Instead, we run only a single phase of a Luby-like MIS algorithm before updating the data structures. Intuitively, this saves a $\log n$ factor because we are running just one phase of a $O(\log n)$-phase algorithm, but it is not clear a priori whether we achieve the same progress as Luby's algorithm in a single phase. We demonstrate that this is indeed the case via more refined treatment of the behavior of each edge.

Our second key modification, which shaves an $\alpha$ factor from the number of rounds, concerns the Luby-like algorithm. Recall that in Luby's algorithm, each vertex $v$ picks a random value $p(v)$ and then joins the MIS if $p(v)$ is the local minimum. In our algorithm, a vertex $v$ instead joins the dominating set if $p(v)$ is an $\alpha$-minimum, which roughly means that $p(v)$ is among the $\alpha$ smallest values that it is compared to. We show that with this relaxed definition, we still have the desired property that no vertex receives more than $\delta\alpha$ votes in a single round.

The main technical challenge is the analysis of the number of rounds. It is tempting to use an analysis similar to that of Luby's algorithm, where we count the expected number of ``removed edges'' over time. However, our above modifications introduce several complications that preclude such an analysis. Instead, we use a carefully chosen function to measure our progress. Throughout the algorithm, we add ``weight'' to particular edges, and our function measures the ``total available weight''. Specifically, whenever a vertex $v$ is added to the dominating set, $v$ adds \emph{weight} to a particular set of edges in its 2-hop neighborhood. We show that the total amount of weight added in a single iteration of the algorithm decreases the total available weight substantially, which allows us to bound the total number of iterations. 

All of our distributed algorithms so far have assumed that $\alpha$ is known to each processor but that $n$ is unknown. We additionally show that all of them can be made to work in the setting where $\alpha$ is unknown but $n$ is known. The idea of this modification is to guess $\log n$ values of $\alpha$ and run a truncated version of the algorithm for each guess. However, it is impossible for an individual processor to know which guess of $\alpha$ is the most accurate without knowing the whole graph, so the processors cannot coordinate their guesses globally. We end up with different processors using different guesses of $\alpha$, but we show that we can nonetheless obtain an algorithm whose approximation factor and running time are in accordance with the correct $\alpha$.

\subsection{Organization}

\cref{sec:pre} is for preliminaries. In \cref{sec:central}, we present our centralized algorithm (\cref{thm:static}). In \cref{sec:dist}, we present our distributed algorithms using MIS: in the $\mathsf{LOCAL}$ model we prove \cref{thm:local}, and in the $\mathsf{CONGEST}$ model we give a randomized algorithm with $O(\alpha^2\log^2 n)$ rounds that serves as a warm-up for the faster algorithm of \cref{thm:improved}. In \cref{sec:imp}, we prove \cref{thm:improved} and show that all of our distributed algorithms can be made to handle unknown $\alpha$.

\section{Preliminaries}\label{sec:pre}
Let $G=(V,E)$ be an unweighted undirected graph. For any $S\subseteq V$, let $G[S]$ be denote the subgraph induced by $S$. For any $v\in V$, $N_G(v)$ denotes the neighborhood of $v$, and $\deg_G(v)=|N_G(v)|$ denotes the degree of $v$. When the graph $G$ is clear from context, we omit the subscript.\\

We define the $\mathsf{LOCAL}$ and $\mathsf{CONGEST}$ models (cf.~\cite{linial1987distributive, linial1992locality,peleg2000distributed}):

\begin{definition}
The $\mathsf{LOCAL}$ model: given a graph $G$ on $n$ vertices, every vertex is a separate processor running one process. Every vertex starts knowing only $n$ and it's own unique identifier. The algorithm works in synchronous rounds, and in every round each vertex performs some computation based on its own current information, then it sends a message to its neighbors, and finally it receives the messages sent to it by its neighbors in that round. 
\end{definition}
\begin{definition}
The $\mathsf{CONGEST}$ model: given a graph $G$ on $n$ vertices, every vertex is a separate processor running one process. Every vertex starts knowing only $n$ and it's own unique identifier. The algorithm works in synchronous rounds, and in every round each vertex performs some computation based on its own current information, then it sends a message to its neighbors  of at most $B = O(\log n)$ bits on each of its edges (possibly a different message to each neighbor), and finally it receives the messages sent to it by its neighbors in that round. 
\end{definition}

For the problem of MDS in both models, the requirement is that at the end of the computation, every vertex knows whether or not it belongs to the dominating set.

The following two claims about graphs of bounded arboricity will be useful. Simple proofs of both can be found in \cite{arikati1997efficient}.

\begin{claim}\label{claim:induced} In a graph of arboricity $\alpha$, every subgraph contains a vertex of degree at most  $2\alpha$.
\end{claim}


\begin{claim}\label{claim:induced2}
In a graph G with arboricity $\alpha$, at least half of the vertices in any subgraph have degree at most $4\alpha$.
\end{claim}


\section{Linear time $O(\alpha)$-approximation for MDS}\label{sec:central}

In this section we will prove \cref{thm:static}, which we recall:

\thmstatic*

\subsection{Algorithm}
A description of our algorithm is as follows. See \cref{alg:static} for the pseudocode. 

We first introduce some notation. Since our algorithm builds off of \cite{jones2013parameterized}, we stick to their notation for the most part. See \cref{fig:dom}. We define a constant $\delta$ and let $\delta\alpha$ be our \emph{degree threshold}. We will set $\delta=2$, but we use the variable $\delta$ so that our analysis also applies to our distributed algorithms, where $\delta$ is a different constant. We maintain a partition of the vertices into three sets: $D$, $B$, and $W$, where initially $D=\emptyset$, $B=\emptyset$, and $W=V$. The set $D$ is our current dominating set, the set $B$ is the vertices not in $D$ with at least one neighbor in $D$, and the set $W$ is the remaining vertices, i.e. the undominated vertices. The set $B$ is further partitioned into two sets based on the degree of each vertex to $W$. Let $B_{low}=\{v\in B : |N(v)\cap W| \leq \delta\alpha\}$ and let $B_{high} = B\setminus B_{low}$. Let $W_{low}=\{v\in W: |N(v)\cap (W\cup B_{high})|\leq \delta\alpha\}$ Also, each vertex $v$ has a \emph{counter} $c_v$ initialized to 0. (The counter $c_v$ counts the  number of ``votes'' that $v$ receives, for the notion of ``votes'' introduced in the technical overview.)

First we claim that while $W$ is nonempty, $W_{low}$ is also nonempty. By \cref{claim:induced}, $G[W\cup B_{high}]$ contains a vertex $v$ of degree at most $2\alpha$. Since $\delta=2$, $v$ cannot be in $B_{high}$ by the definition of $B_{high}$, so $v$ must be in $W$, and hence in $W_{low}$.

The algorithm proceeds as follows. While there still exists an undominated vertex (i.e. while $W\not=\emptyset$), we do the following. 
First, we pick an arbitrary vertex $w\in W_{low}$ (we have shown that $W_{low}$ is nonempty). 
Then, for all $v\in N(w)\cap (W\cup B_{high})$, we increment $c_v$, and if $c_v=\delta\alpha$, then we add $v$ to $D$. Then, we add $w$ to $D$. Lastly, we update the sets $B$, $_{low}$, $B_{high}$, $W$, and $W_{low}$ according to their definitions. This concludes the description of the algorithm.

\begin{algorithm}
\caption{Linear time $O(\alpha)$-approximation for MDS}\label{alg:static}

\begin{algorithmic}[1]
\State Initialize partition: $D\gets \emptyset$, $B=\emptyset$, $B_{high}\gets\emptyset$, $B_{low}\gets\emptyset$, $W\gets V$, $W_{low}=\{v\in V: \deg(v)\leq \delta\alpha\}$
\State Initialize counters: $\forall v\in V: c_v \leftarrow 0$
\While {$W \neq \phi $}
\State $w\gets$ a vertex in $W_{low}$ \label{line:w}
\ForAll {$v \in N(w)\cap(W\cup B_{high})$} \label{line:foreach}
\State $c_v \leftarrow c_v + 1$\label{line:counter}
\If{$c_v=\delta\alpha$} 
\State $D \leftarrow D\cup{v}$ \label{line:passive}
\EndIf
\EndFor
\State $D \leftarrow D\cup{w}$\label{line:active}
\LineComment{{\bf Bookkeeping to update partition:}}
\State $B=\{v: N(v)\cap D\not=\emptyset\}$\label{line:b}
\State $B_{low}=\{v\in B: |N(v)\cap W| \leq \delta\alpha\}  $
\State $B_{high} = B\setminus B_{low}$
\State $W=V\setminus (D\cup B)$\label{line:ww}
\State{$W_{low}=\{v\in W: |N(v)\cap (W\cup B_{high})|\leq \delta\alpha\}$}
\EndWhile
\State Return D

\end{algorithmic}
\end{algorithm}

\subsection{Analysis}
First, we note that $D$ is indeed a dominating set because the algorithm only terminates once the set $W$ of vertices that are not dominated, is empty.
\subsubsection{Approximation ratio analysis}\label{sec:approx}
Let $OPT$ be an optimal MDS. We will prove that the set $D$ returned by \cref{alg:static} is of size at most $4\delta\alpha\cdot|OPT|$. 

We first make the following simple claim about the behavior of the partition of vertices over time.
\begin{restatable}{claim}{claimenter}
\label{claim:enter}$ $
\begin{enumerate}
\item No vertex can ever leave $D$.
\item No vertex can ever enter $W$ from another set.
\item No vertex can ever leave $B_{low}$. 
\end{enumerate}
\end{restatable}
\begin{proof}
Item 1 is by definition. Item 2 follows from item 1 combined with the fact that $W$ is defined as the set of vertices with no neighbors in $D$. Now we prove item 3. A vertex from $B_{low}$ cannot enter $W$ by item 2. A vertex from $B_{low}$ cannot enter $B_{high}$ since the degree partition of $B$ is based on degree to $W$, and by item 2 the degree of any vertex to $W$ can only decrease over time. A vertex from $B_{low}$ cannot enter $D$ because there are two ways a vertex can enter $D$: on \cref{line:passive} a vertex can only enter $D$ from $W\cup B_{high}$, and on \cref{line:active} a vertex can only enter $D$ from $W$. 
\end{proof}

In order to show that $|D|\leq 4\delta\alpha\cdot|OPT|$, we partition $D$ into two sets, $D_{active}$ and $D_{passive}$, and bound each of these sets separately. The set $D_{active}$ consists of the vertices added to $D$ due to being chosen as the vertex $w$; that is, the vertices added to $D$ in \cref{line:active} of \cref{alg:static}. The set $D_{passive}$ consists of the vertices added to $D$ as a result of their counters reaching $\delta\alpha$; that is, the vertices added to $D$ in \cref{line:passive} of \cref{alg:static}. We will first bound $|D_{active}|$. 

\begin{claim}\label{claim:active} $|D_{active}|\leq 2\delta\alpha\cdot|OPT|$.
\end{claim}
\begin{proof}

For each vertex $v \in D_{active}$, we assign $v$ to an arbitrary vertex $u\in N(v)\cap OPT$, and we say that $u$ \emph{witnesses} $v$. Such a vertex $u$ exists since $OPT$ is a dominating set. For each vertex $u\in OPT$, let $D_u\subseteq D_{active}$ be the set of vertices that $u$ witnesses. Our goal is to show that for each $u\in OPT$, $|D_u|\leq 2\delta\alpha$.

Fix a vertex $u\in OPT$. We partition the vertices $v\in D_u$ into two sets $D_u[B_{low}]$ and $D_u[B_{high}\cup W]$. Let $D_u[B_{low}]\subseteq D_u$ be the vertices that enter $D$ while $u$ is in $B_{low}$. Let $D_u[B_{high}\cup W]\subseteq D_u$ be vertices that enter $D$ while $u$ is in $B_{high}\cup W$. We note that no vertex in $D_u$ can enter $D$ while $u$ is in $D$, because by definition, every vertex in $D_{active}\supseteq D_u$ moves directly from $W$ to $D$. Therefore, $D_u=D_u[B_{low}]\cup D_u[B_{high}\cup W]$.

We first bound $\big|D_u[B_{low}]\big|$. By definition, while $u$ is in $B_{low}$, $u$ has at most $\delta\alpha$ neighbors in $W$. Since no vertex can ever enter $W$ by \cref{claim:enter}, no vertex can ever enter $N(u)\cap W$. Therefore, starting from the time that $u$ first enters $B_{low}$, the total number of vertices ever in $N(u)\cap W$ is at most $\delta\alpha$. Every vertex $v\in D_u[B_{low}]$ is in $N(u)\cap W$ right before moving to $D$, so $\big|D_u[B_{low}]\big|\leq \delta\alpha$.

Now, we bound $D_u[B_{high}\cup W]$. By the specification of the algorithm, whenever a vertex $v\in D_u[B_{high}\cup W]$ enters $D$, the counter $c_u$ is incremented. Once $c_u$ reaches $\delta\alpha$, $u$ is added to $D$. Therefore, $\big|D_u[B_{high}\cup W]\big|\leq \delta\alpha$.

Putting everything together, we have $|D_u|=\big|D_u[B_{low}]\big|+\big|D_u[B_{high}\cup W]\big|\leq 2\delta\alpha$.
\end{proof}

Now we bound $D_{passive}$.
\begin{claim}\label{claim:passive}
$|D_{passive}|\leq|D_{active}|$.

\end{claim}
\begin{proof}
We will show that every vertex in $D_{passive}$ has at least $\delta\alpha$ neighbors in $D_{active}$, while every vertex in $D_{active}$ has at most $\delta\alpha$ neighbors in $D_{passive}$. Then, by the pigeonhole principle, it follows that $|D_{passive}|\leq|D_{active}|$.

First, we will show that every vertex in $D_{passive}$ has at least $\delta\alpha$ neighbors in $D_{active}$. By definition, every vertex $v\in D_{passive}$ has had its counter $c_v$ incremented $\delta\alpha$ times. Every time $c_v$ is incremented, one of $v$'s neighbors (the vertex $w$ from \cref{alg:static}) is added to $D$, joining $D_{active}$. Each such neighbor of $v$ that joins $D_{active}$ is distinct since every vertex can be added to $D$ at most once by \cref{claim:enter}. Therefore, every vertex in $D_{passive}$ has at least $\delta\alpha$ neighbors in $D_{active}$. 

Now we will show that every vertex in $D_{active}$ has at most $\delta\alpha$ neighbors in $D_{passive}$. Fix a vertex $w\in D_{active}$. By definition, when $w$ enters $D$, $w$ is moved straight from $W$ to $D$. Thus, by \cref{claim:enter}, $w$ is never in $B$. Therefore, $w$ is added to $D$ before any of its neighbors are added to $D$, as otherwise $w$ would enter $B$. Therefore, when $w$ enters $D$, all of $w$'s neighbors that will enter $D_{passive}$ are in $B\cup W$. By \cref{claim:enter}, no vertex in $B_{low}$ can ever enter $D$, so actually, when $w$ enters $D$ all of $w$'s neighbors that will enter $D_{passive}$ are in $B_{high}\cup W$. By definition, when $w$ enters $D$, $w$ has at most $\delta\alpha$ neighbors in $B_{high}\cup W$. 
Therefore, $w$ has at most $\delta\alpha$ neighbors in $D_{passive}$.
\end{proof}

Combining \cref{claim:active} and \cref{claim:passive}, we have that $|D|=|D_{active}|+|D_{passive}|\leq 4\delta\alpha\cdot |OPT|$.

\subsubsection{Running time analysis}
Our goal is to prove that \cref{alg:static} runs in $O(m)$ time.

Throughout the execution of the algorithm, we maintain a data structure that consists of the following:
\begin{itemize}

    \item The partition of $V$ into $D$, $B$, $W$; with subsets $B_{low}$, $B_{high}$, $W_{low}$
    \item The induced graph $G[W\cup B_{high}]$ represented as an adjacency list
    \item For each vertex $v\in W\cup B_{high}$, the quantities $|N(v)\cap W|$ and $|N(v)\cap (W\cup B_{high})|$
\end{itemize}


First, we show that the data structure can be initialized in $O(m)$ time. Initially $D\cup B_{high}\cup B_{low}=\emptyset$, $W=V$, and the induced graph $G[W\cup B_{high}]=G$. For every vertex $v\in V$, initially $|N(v)\cap W|=|N(v)\cap (W\cup B_{high})|=\deg(v)$. Initially $W_{low}=\{v\in V : \deg(v)\leq \delta\alpha\}$.

Now, we show that the data structure can be maintained in $O(m)$ time. In particular, we will show that to maintain this data structure, it suffices to scan the neighborhood of a vertex every time it either leaves $W\cup B_{high}$ (and enters $B_{low}\cup D$), enters $D$, or leaves $W$. Note that by \cref{claim:enter}, each of these events only happens once per vertex. As a consequence, the total amount of time spent scanning neighborhoods is $O(m)$.

We assume inductively that we have maintained the data structure so far, and we consider the next iteration of the {\bf for each} loop. 
First, we consider maintenance of the partition of $V$ into $D$, $B_{low}$, $B_{high}$, and $W$. During an iteration, the only changes made to the partition are the addition of at least one vertex to $D$ (on \cref{line:passive} and/or \cref{line:active}), and the resulting update of the rest of the partition. To maintain the partition we do the following. When we add a vertex $v$ to $D$, we remove $v$ from whichever set it was previously in. Then, we update $B$ by scanning $N(v)$ and adding every vertex $u\in N(v)\setminus D$ to $B$, removing $u$ from whichever set it was previously in. Updating $D$ and $B$ automatically updates $W$ since $W=V\setminus (D\cup B)$. Before updating $B_{low}$ and $B_{high}$, we first need to update $|N(v)\cap W|$. To do this, whenever a vertex $v$ leaves $W$, we scan $N(v)$ and for each $u\in N(v)$, we decrement $|N(u)\cap W|$. Whenever we decrement $|N(u)\cap W|$ down to $\delta\alpha$ for a vertex $u\in B_{high}$, we move $u$ to $B_{low}$. This concludes the maintenance of the partition of $V$ into $D$, $B_{low}$, $B_{high}$, and $W$.

It remains to update $G[W\cup B_{high}]$, $|N(v)\cap(W\cup B_{high})|$, and $W_{low}$.
Whenever we remove a vertex $v$ from $W\cup B_{high}$, we scan $N(v)$ and for each vertex $u\in N(v)$, we remove the edge $(u,v)$ from $G[W\cup B_{high}]$ and decrement $|N(u)\cap (W\cup B_{high})|$. Whenever we decrement $|N(u)\cap (W\cup B_{high})|$ down to $\delta\alpha$ for $u\in W$, we add $u$ to $W_{low}$. This concludes the running time analysis for maintaining the data structure.

Now we will show that maintaining the data structure allows the algorithm to run in time $O(m)$. First, each iteration of the {\bf while} loop adds at least one vertex to $D$ (on \cref{line:active}), and by \cref{claim:enter}, each vertex is added to $D$ at most once, so the total number of iterations of the {\bf while} loop is at most $n$. Now we will go line by line through the body of the {\bf while} loop. On \cref{line:w}, we let $w$ be a vertex in $W_{low}$. On \cref{line:foreach}, we loop through every vertex in $|N(w)\cap (W\cup B_{high})|$. The number of iterations of this loop is at most $\delta\alpha$ by choice of $w$. Furthermore, identifying all of the vertices to loop through takes time $O(\alpha)$ since our data structure explicitly maintains $G[W\cup B_{high}]$. In \cref{line:counter} through \cref{line:active}, we update counters and then add vertices to $D$, which takes constant time per iteration of the loop. In \cref{line:b} through \cref{line:ww} we update $B$, $B_{low}$, $B_{high}$, and $W$, which takes constant time since we store these sets in our data structure. Thus, given access to the data structure, the algorithm runs in time $O(n\alpha)=O(m)$. 

Previously we showed that maintaining the data structure takes time $O(m)$, so we have that the entire algorithm takes time $O(m)$.








\section{Distributed $O(\alpha)$-approximation for MDS using MIS}\label{sec:dist}

In this section we will prove \cref{thm:local}, which we recall:

\thmlocal*

In this section we also show how to modify of the proof of \cref{thm:local} to get a bound in the $\mathsf{CONGEST}$ model:

\begin{restatable}{theorem}{thmcongest}
\label{thm:congest}
For graphs of arboricity $\alpha$, there is a randomized distributed algorithm in the $\mathsf{CONGEST}$ model that gives an $O(\alpha)$-approximation for MDS that runs in $O(\alpha^2\log^2 n)$ rounds with high probability. The algorithm works even when either $\alpha$ or $n$ is unknown to each processor.
\end{restatable}

In \cref{sec:imp}, we will use the algorithm of \cref{thm:congest} as a starting point to get an improved algorithm with  $O(\alpha\log n)$ rounds. 

The algorithms presented in this section assume that $\alpha$ is known to each processor but $n$ is unknown. We defer discussion of handling unknown $\alpha$ to the next section.


\subsection{Algorithm}
\paragraph{Overview} Our algorithm is an adaptation of our centralized algorithm from \cref{thm:static} to the distributed setting. Recall that in our centralized algorithm, we repeatedly choose a vertex $w\in W_{low}$, add $w$ to the dominating set, and increment the \emph{counter} of $w$'s neighbors that are in $W\cup B_{high}$. For our distributed algorithms, we would like process \emph{many} vertices in $W_{low}$ in parallel. There are in fact many vertices in $W_{low}$ (if $\delta\geq 4$) since \cref{claim:induced2} implies that at least half of the vertices in any subgraph has degree at most $4\alpha$. However, we cannot simply process all of $W_{low}$ at once. In particular, if a vertex $v$ has many neighbors being processed in parallel, $v$ might have its counter incremented once for each of these neighbors. This is undesirable because the analysis of our centralized algorithm relies on the fact that once a vertex has its counter incremented to $\delta\alpha$, it is added to the dominating set. Therefore, we would like to guarantee that only a limited number of $v$'s neighbors are processed in parallel. 

This is where the MIS problem becomes relevant: we ensure that no vertex has more than one neighbor being processed in parallel by taking an MIS $I$ with respect to the graph $G_{low}$ defined as follows: the vertex set of $G_{low}$ is $W_{low}$. There is an edge $(u,v)$ in $G_{low}$ if there is a path of length 2 between $u$ and $v$ in $G[W\cup B_{high}]$. Note that because no vertex has more than one neighbor in $I$, we can process all vertices in $I$ in parallel and only increase the counter of each vertex by at most one. 

The algorithms for \cref{thm:local} and \cref{thm:congest} are identical except for the MIS subroutine. \cref{thm:local} is for the $\mathsf{LOCAL}$ model so we can simply run any distributed MIS algorithm that works in the $\mathsf{LOCAL}$ model on $G_{low}$ as a black box. On the other hand, \cref{thm:congest} is for the $\mathsf{CONGEST}$ model and because $G_{low}$ can have higher degree than $G$, running an MIS algorithm directly on $G_{low}$ could result in messages that become too large after translating the algorithm to run on $G$. To bypass this issue, we use a simple modification of Luby's algorithm that computes $I$ using only small messages, without increasing the number of rounds.

\paragraph{Algorithm description} We provide a description of the algorithms here, and include the pseudocode in \cref{alg:dist}. The only difference between the algorithms for \cref{thm:local} and \cref{thm:congest} is the MIS subroutine, which we will handle separately later.

The sets $D$, $B$, $W$, $B_{high}$, $B_{low}$, and $W_{low}$ are defined exactly the same as in our centralized algorithm, except we set $\delta=4$ instead of $\delta=2$ so that we can apply \cref{claim:induced2} instead of \cref{claim:induced}. We repeat the definitions here for completeness. The set $D$ is our current dominating set, the set $B$ is the vertices not in $D$ with at least one neighbor in $D$, and the set $W$ is the remaining vertices, i.e. the undominated vertices. The set $B$ is further partitioned into two sets based on the degree of each vertex to $W$. Let $B_{low}=\{v\in B : |N(v)\cap W| \leq \delta\alpha\}$ and let $B_{high} = B\setminus B_{low}$. Also, let $W_{low}=\{v\in W: |N(v)\cap (W\cup B_{high})|\leq \delta\alpha\}$. Lastly, each vertex $v$ has a \emph{counter} $c_v$.

Each vertex $v$ maintains the following information:
\begin{itemize}
    \item The set(s) among $D$, $B$, $W$, $B_{high}$, $B_{low}$, and $W_{low}$ that $v$ is a member of.
    \item The quantity $|N(v)\cap W|$.
    \item The quantity $|N(v)\cap (W\cup B_{high})|$.
    \item The counter $c_v$. 
\end{itemize}

At initialization, every vertex $v$ is in $W$ (so $D$ and $B$ are empty). Consequently, the quantities $|N(v)\cap W|$ and $|N(v)\cap (W\cup B_{high})|$ are both equal to $\deg(v)$. For each vertex $v$, if $\deg(v)\leq \delta\alpha$, then $v\in W_{low}$. Each counter $c_v$ is initialized to 0.

It will be useful to define the graph $G_{low}$, which changes over the execution of the algorithm:

\begin{definition}
Let $G_{low}$ be the graph with vertex set $W_{low}$ such that there is an edge $(u,v)$ in $G_{low}$ if there is a path of length 2 between $u$ and $v$ in $G[W\cup B_{high}]$. 
\end{definition}

The algorithm proceeds as follows. Repeat the following until $W$ is empty. Compute an MIS $I$ with respect to $G_{low}$. This step is implemented differently for \cref{thm:local} and \cref{thm:congest}, and we describe the details of this step later.

Then, each vertex in $I$ adds itself to $D$ and tells its neighbors to increment their counters. Whenever the counter of a vertex reaches $\delta\alpha$, it enters $D$ (and does \emph{not} tell its neighbors to increment their counters).

Whenever a vertex moves from one set of the partition to another, it notifies each of its neighbors $v$ so that $v$ can update the quantities $|N(v)\cap W|$ and $|N(v)\cap (W\cup B_{high})|$, and move to the appropriate set. When no more vertices are left in $W$, $B_{high}$ is also empty, and all processors terminate. This concludes the description of the algorithm. See \cref{alg:dist} for the precise ways that vertices react to the messages that they receive.

\begin{algorithm}

\caption{Distributed $O(\alpha)$-approximation for MDS using MIS}\label{alg:dist}

\noindent\begin{minipage}{\textwidth}
\renewcommand{\thempfootnote}{\arabic{mpfootnote}}
\begin{algorithmic}[1]
\State{Initialize partition: $D\gets \emptyset$, $B_{high}\gets\emptyset$, $B_{low}\gets\emptyset$, $W\gets V$, $W_{low}\gets \{v\in V: \deg(v)\leq \delta\alpha\}$}
\State{Initialize counters: $\forall v\in V: c_v \leftarrow 0$}
\State{Initialize degrees: $\forall v\in V: |N(v)\cap W|=\deg(v)$, $|N(v)\cap (W\cup B_{high})|=\deg(v)$}

\While{$W\not=\emptyset$}
\State{Find an MIS $I$ with respect to the graph $G_{low}$}
\State{Each vertex $v$ runs the following procedure:}

    \If{$v\in I$}
        \State{Move $v$ to $D$}
        \State{{\bf Send} \textsc{increment counter} message to neighbors}
        \State{{\bf Send} \textsc{moved from $W$ to $D$} message to neighbors}
    \EndIf

\If{$v\in W\cup B_{high}$ and $v$ receives \textsc{increment counter}} \label{line:book}
    \State{Increment $c_v$}
    \If{$c_v=\delta\alpha$}
        \If{$v\in W$}
            \State{{\bf Send} \textsc{moved from $W$ to $D$} message to neighbors}
            \EndIf
        \If{$v\in B_{high}$}
            \State{{\bf Send} \textsc{moved from $B_{high}$ to $D$} message to neighbors}
            \EndIf
            \State{Move $v$ to $D$}\label{line:movecount}
            \EndIf
            \EndIf
\LineComment{{\bf The rest of the algorithm is bookkeeping}}
\If{$v$ receives \textsc{moved from $W$ to $D$}}
    \State{Decrement $|N(v)\cap W|$}
    \If{$v\in B_{high}$ and $|N(v)\cap W|=\delta\alpha$}
    \State{Move $v$ to $B_{low}$}
    \EndIf
    \EndIf

\If{$v$ receives \textsc{moved from $W$ to $D$} or \textsc{moved from $B_{high}$ to $D$}}
    \State{Decrement $|N(v)\cap (W\cup B_{high})|$}
    \If{$v\in W$ and $|N(v)\cap W|\leq \delta\alpha$}
        \State{Move $v$ to $B_{low}$}
        \State{{\bf Send} \textsc{moved from $W$ to $B_{low}$} message to neighbors}
    \ElsIf{$v\in W$ and $|N(v)\cap W|> \delta\alpha$}
        \State{Move $v$ to $B_{high}$}
        \State{{\bf Send} \textsc{moved from $W$ to $B_{high}$} message to neighbors}
        \EndIf
\EndIf

\If{$v$ receives \textsc{moved from $W$ to $B_{low}$} or \textsc{moved from $W$ to $B_{high}$}}
    \State{Decrement $|N(v)\cap W|$}
    \If{$v\in B_{high}$ and $|N(v)\cap W|=\delta\alpha$}
    \State{Move $v$ to $B_{low}$}
    \EndIf
    \EndIf

    \If{$v$ receives \textsc{moved from $W$ to $B_{low}$}}
    \State{Decrement $|N(v)\cap (W\cup B_{high})|$}
    \If{$v\in W$ and $|N(v)\cap (W\cup B_{high})|=\delta\alpha$}
        \State{Add $v$ to $W_{low}$}
        \EndIf
        \EndIf
          \EndWhile
\end{algorithmic}
\end{minipage}
\end{algorithm}

\paragraph{MIS subroutine}
\cref{thm:local} is a reduction from MDS to MIS, while \cref{thm:congest} is not, so we need to describe the MIS subroutine (in the $\mathsf{CONGEST}$ model) only for \cref{thm:congest}. Recall that we cannot use a reduction to MIS in the $\mathsf{CONGEST}$ model because running an MIS algorithm directly on $G_{low}$ could result in messages that become too large after translating the algorithm to run on $G$.


Our goal is to compute an MIS with respect to $G_{low}$, using small messages sent over $G$.
We use a simple adaptation of Luby's algorithm. Recall that Luby's algorithm builds an MIS $I$ as follows. While the graph is non-empty, do the following: Add all singletons to $I$. Then, each vertex $v$ picks a random value $p(v)\in[0,1]$. Then, all vertices whose value is less than that of all of their neighbors are added to $I$. Then, all vertices that are in $I$ or have a neighbor in $I$ are removed from the graph for the next iteration of the loop.

We use the following adaptation of Luby's algorithm. See \cref{alg:lubymod} for the pseudocode. Initially, the set $L$ of \emph{live} vertices is the set $W_{low}$. While $L\not=\emptyset$, do the following: Each vertex $v\in L$ picks a random value $p(v)\in[0,1]$. In the first round each $v\in L$ sends $p(v)$ to its neighbors. In the second round, each vertex that receives one or more values $p(v)$, forwards to its neighbors the minimum value that it received. Then, for each vertex $v\in W_{low}$, if $p(v)$ is equal to the minimum value that $v$ receives in the second round, $v$ is added to $I$. When $v$ is added to $I$, $v$ notifies its neighbors, and each neighbor of $v$ that is in $W\cup B_{high}$ forwards this notification to their neighbors. Note that each vertex has at most one neighbor in $I$, so forwarding this notification only takes one round. Now, every vertex knows whether it has a neighbor with respect to $G_{low}$ that is in $I$, and every vertex that does is removed from $L$ for the next iteration of the loop.

The proof that this algorithm runs in $O(\log n)$ rounds with high probability and produces an MIS with respect to $G_{low}$ is the same as the analysis of Luby's algorithm and we will not include it here.

\begin{algorithm}
\caption{Distributed MIS with respect to $G_{low}$ in the $\mathsf{CONGEST}$ model}\label{alg:lubymod}

\begin{algorithmic}[1]
\State $L=W_{low}$
\While{$L\not=\emptyset$}
\State Each vertex $v$ runs the following procedure:
\If{$v\in L$}
\State $p(v)\gets$ a value in $[0, 1]$ chosen uniformly at random
\State {\bf Send} $p(v)$ message to neighbors
\EndIf
\State {\bf Send} $m_v=\min_{y\in N(v)\cap L} p(y)$ message to neighbors
\If{$p(v)=\min_{y\in N(v)} m_y$}
\State Add $v$ to $I$
\State {\bf Send} \textsc{added} message to neighbors
\EndIf
\If{$v\in W\cup B_{high}$ and $v$ receives \textsc{added}}
\State {\bf Send} \textsc{neighbor added} message to neighbors
\EndIf
\If{$v$ receives \textsc{neighbor added} and $v\in L$}
\State Remove $v$ from $L$
\EndIf
\EndWhile
\end{algorithmic}
\end{algorithm}

\subsection{Analysis}

The proof that \cref{alg:dist} achieves an $O(\alpha)$-approximation is precisely the same as that of the centralized algorithm (see \cref{sec:approx}) given that no counter $c_v$ ever exceeds $\delta\alpha$. This is true because in a single iteration of the {\bf while} loop each vertex can only have its counter incremented once since only vertices in the MIS $I$ send \textsc{increment counter} messages, and each vertex in $W\cup B_{high}$ only has at most one neighbor in $I$. This bound on the number of neighbors in  $I$ holds, since otherwise there is a path of length 2 between two vertices in $G[ W\cup B_{high}]$, making $I$ not an independent set in $G_{low}$. Once $c_v$ reaches $\delta\alpha$, the vertex $v$ enters $D$, which prevents $c_v$ from increasing in the future.

Our goal in this section is to prove that if the MIS subroutine takes $R(n)$ rounds, then \cref{alg:dist} takes $O(R(n)\cdot\alpha^2\log n)$ rounds. First, we note that the body of the {\bf while} loop besides the MIS subroutine takes a constant number of rounds. Thus, our goal is to show that the number of iterations of the {\bf while} loop is $O(\alpha^2\log n)$.

We begin with a simple claim about the behavior of the partition of vertices over time:

\begin{restatable}{claim}{claimenterr}
\label{claim:enter2}$ $
\begin{enumerate}
\item No vertex can ever enter $W$ from another set.
\item No vertex can ever move from $W_{low}$ to $W_{high}$.
\end{enumerate}
\end{restatable}
\begin{proof}
The proof of item 1 is the same as in the proof of \cref{claim:enter}. For item 2, it is impossible for a vertex to move from $W_{low}$ to $W_{high}$ since for all $v$ the quantity $N(v)\cap (W\cup B_{high})$ that determines membership in $W_{low}$ versus $W_{high}$, can only decrease over time (in \cref{alg:dist}, this quantity is only decremented).
\end{proof}

We begin with the following claim, which when combined with \cref{claim:enter2}, implies that each vertex only spends a limited number of rounds in $W_{low}$. 

\begin{claim}\label{claim:wlow} For every vertex $v$ that is ever in $W_{low}$, within $(\delta\alpha)^2$ iterations of the {\bf while} loop after $v$ joins $W_{low}$, $v$ leaves $W$.
\end{claim}

\begin{proof}
First we note that by \cref{claim:enter2} no vertex can ever move from $W_{low}$ to $W_{high}$. Thus, if $v$ is in $W_{low}$, $v$ will remain in $W_{low}$ until $v$ leaves $W$. Suppose $v$ is in $W_{low}$ at the beginning of an iteration of the {\bf while} loop. Because $I$ is an MIS with respect to $G_{low}$, if $v$ does not join $I$ during this iteration, then $v$ has a neighbor $y\in W\cup B_{high}$ such that a neighbor $z$ of $y$ joins $I$. As a result, $z$ immediately joins $D$ and $c_y$ is incremented. Thus, during every iteration that $v$ remains in $W_{low}$, a vertex in $N(v) \cap (W\cup B_{high})$ has its counter incremented. Recall that whenever a vertex has its counter incremented $\delta\alpha$ times, it joins $D$. Because $v\in W_{low}$, we have that $|N(v) \cap (W\cup B_{high})|\leq \delta\alpha$. Therefore, the event that a vertex in $N(v) \cap (W\cup B_{high})$ has its counter incremented can only happen at most $(\delta\alpha)^2$ times. Thus, $v$ can only remain in $W_{low}$ for $(\delta\alpha)^2$ iterations of the {\bf while} loop.
\end{proof}

We will complete the analysis using the fact that enough vertices are in $W_{low}$ at any given point in time. In particular, \cref{claim:induced2} implies that at least half of the vertices in $W\cup B_{high}$ are in $W_{low}$. This implies that at least half of the vertices in $W$ are in $W_{low}$. Formally, we divide the execution of the algorithm into phases where each phase consists of $(\delta\alpha)^2$ iterations of the {\bf while} loop. At the beginning of any phase, at least half of the vertices in $W$ are in $W_{low}$. By the end of the phase, all of these vertices have left $W$ by \cref{claim:wlow}. Therefore, each phase witnesses at least half of the vertices in $W$ leaving $W$. By \cref{claim:enter2}, no vertex can re-enter $W$, so there can only be $O(\log n)$ phases.

Putting everything together, there are $O(\log n)$ phases, each consisting of $(\delta\alpha)^2$ iterations of the {\bf while} loop, and one iteration of the {\bf while} loop takes $O(R(n))$ rounds. Therefore, the total number of rounds is $O(R(n)\cdot\alpha^2\log n)$.

For \cref{thm:congest}, $R(n)=O(\log n)$, so the number of rounds is $O(\alpha^2\log^2 n)$.

\section{Faster Randomized Distributed $O(\alpha)$-approximation for MDS}\label{sec:imp}

In this section we will prove \cref{thm:improved}, which we recall:

\thmimproved*

We first present an algorithm that assumes that $\alpha$ is known to each processor but $n$ is unknown. In Section~\ref{sec:unknown}, we handle the case of unknown $\alpha$.

\subsection{Algorithm}
\paragraph{Overview} We use our $O(\alpha^2 \log^2 n)$ round algorithm from \cref{thm:congest} as a starting point (though the algorithm description and analysis are self-contained). Our goal is to shave both a $\log n$ factor and an $\alpha$ factor from the number of rounds. To do so, we use a combination of two key modifications, which respectively address the two factors that we wish to shave. 

Our first key modification, which shaves a $\log n$ factor from the number of rounds, is that instead of using a Luby-style algorithm as a black box, we open the box and run only one phase of a Luby-style algorithm at a time. Here, one phase means that each participating vertex $v$ picks a single random value $p(v)$ and enters $D$ if $p(v)$ is a local minimum. Between each such phase, we update the dominating set $D$ as well as the information stored by each vertex. This way, we can embed the analysis of the Luby-style algorithm into our analysis instead of repeatedly paying for for all $\log n$ phases of a black-box algorithm.

Our second key modification, which shaves an $\alpha$ factor from the number of rounds, is that instead of adding $v$ to $D$ only when $p(v)$ is the single local minimum, we allow $v$ to be added to $D$ when $p(v)$ is an \emph{$\alpha$-minimum}. The definition of an \emph{$\alpha$-minimum} is slightly nuanced due to the fact that we need to be able to compute it using small messages, but it roughly means that $p(v)$ is among the $\alpha$ smallest values that it is compared to. Using this modification we can ensure that during each iteration of our algorithm each vertex only has its counter incremented by $O(\alpha)$. Even though each vertex in our previous $O(\alpha^2 \log^2 n)$-round algorithm only had its counter incremented by at most 1 during each iteration, this change does not asymptotically increase the approximation factor.

The main technical part of the argument is the probabilistic analysis of the number of rounds. We would like to use an analysis similar to that of Luby's algorithm, however there are a few obstacles. Recall that to analyze Luby's algorithm, one can argue that after a single phase, a constant fraction of the edges in the graph are removed in expectation. Our first obstacle is that we are running a phase of a Luby-style algorithm on an auxiliary graph that is different from our original graph; in particular, an edge in the auxiliary graph can represent a 2-hop path in the original graph, and it is not clear how removing an edge from the auxiliary graph translates to the original graph. That is, if a constant fraction of edges are removed in the auxiliary graph, this doesn't necessarily mean that a constant fraction of edges in the original graph are removed. A second obstacle is that we need a more nuanced notion than ``removing an edge'' as in Luby's algorithm since due to our second modification, up to $\alpha$ vertices could all affect the same edge simultaneously. To address these obstacles, we use a carefully chosen function to measure our progress. Throughout the algorithm, we add ``weight'' to particular edges, and our function measures the ``total available weight''. Specifically, whenever a vertex $v$ is added to the dominating set, $v$ adds \emph{weight} to a particular set of edges in its 2-hop neighborhood. We show that the total amount of weight added in a single iteration of the algorithm decreases the expected total available weight substantially, which allows us to bound the total number of iterations. 


\paragraph{Algorithm Description} We include a description of the algorithm here, and include the pseudocode in \cref{alg:distimp}.

The partition of the vertices is exactly the same as in \cref{alg:dist}, with the addition of the set $W_{high}$. We repeat all of the definitions for completeness. The set $D$ is our current dominating set, the set $B$ is the vertices not in $D$ with at least one neighbor in $D$, and the set $W$ is the remaining vertices, i.e. the undominated vertices. The set $B$ is further partitioned into two sets based on the degree of each vertex to $W$. Let $B_{low}=\{v\in B : |N(v)\cap W| \leq \delta\alpha\}$, where $\delta=4$, and let $B_{high} = B\setminus B_{low}$. Also, let $W_{low}=\{v\in W: |N(v)\cap (W\cup B_{high})|\leq \delta\alpha\}$. We additionally define $W_{high}=W\setminus W_{low}$. Lastly, each vertex $v$ has a \emph{counter} $c_v$.

Each vertex $v$ maintains the following information:
\begin{itemize}
    \item The set(s) among $D$, $B$, $W$, $B_{high}$, $B_{low}$, $W_{high}$, and $W_{low}$ that $v$ is a member of.
    \item The quantity $|N(v)\cap W|$.
    \item The quantity $|N(v)\cap (W\cup B_{high})|$.
    \item The counter $c_v$. 
\end{itemize}

At initialization, every vertex $v$ is in $W$ (so $D$ and $B$ are empty). Consequently, the quantities $|N(v)\cap W|$ and $|N(v)\cap (W\cup B_{high})|$ are both equal to $\deg(v)$. For each vertex $v$, if $\deg(v)\leq \delta\alpha$, then $v\in W_{low}$. Each counter $c_v$ is initialized to 0.

It will be useful to define the graph $G_{bi}\subseteq G$ that changes over the course of the execution of the algorithm:

\begin{definition}$G_{bi}$ is a bipartite graph on the vertex set $B_{high}\cup W$. One side of the bipartition is $B_{high}\cup W_{high}$ and the other side is $W_{low}$. The edge set of $G_{bi}$ is the set of edges in $G$ with one endpoint in each side of the bipartition. 
\end{definition}

The algorithm proceeds as follows. Repeat the following until $W$ is empty. In the first round, each vertex $v\in W_{low}$ picks a value $p(v)\in [0,1]$ uniformly at random and sends $p(v)$ to its neighbors. The next step is for $v$ to determine whether $p(v)$ is an \emph{$\alpha$-minimum}; $p(v)$ is said to be an \emph{$\alpha$-minimum} if for every $u\in N_{G_{bi}}(v)$, $p(v)$ is among the $\alpha$ smallest values of vertices in $N_{G_{bi}}(u)$.
To determine which values are $\alpha$-minima, in the second round each vertex $u\in B_{high}\cup W_{high}$ sends \textsc{ack} to each vertex $v\in N_{G_{bi}}(u)$ such that $p(v)$ is among the $\alpha$ smallest values that $u$ received.
If $v\in W_{low}$ receives \textsc{ack} from all $u\in N_{G_{bi}}(v)$, then $v$ is added to $D$ (if $N_{G_{bi}}(v)$ is empty then $v$ is added to $D$) and $v$ tells its neighbors to increment their counters (in the third round). Whenever the counter $c_u$ of a vertex $u$ reaches $\delta\alpha$, $u$ enters $D$. (Note that if $u$ enters $D$ as a result of $c_u$ reaching $\delta\alpha$, $c_u$ does \emph{not} tell its neighbors to increment their counters.)

Whenever a vertex $v$ moves from one set of the partition to another, $v$ notifies each vertex $u\in N(v)$ so that $u$ can update the quantities $|N(u)\cap W|$ and $|N(u)\cap (W\cup B_{high})|$, and move to the appropriate set. The bookkeeping for updating this information is identical to that of \cref{alg:dist}. When no more vertices are left in $W$, $B_{high}$ is also empty, and all processors terminate. This concludes the description of the algorithm. 

\begin{algorithm}

\caption{Faster Randomized Distributed $O(\alpha)$-approximation for MDS}\label{alg:distimp}

\noindent\begin{minipage}{\textwidth}
\renewcommand{\thempfootnote}{\arabic{mpfootnote}}
\begin{algorithmic}[1]
\State{Initialize partition: $D\gets \emptyset$, $B_{high}\gets\emptyset$, $B_{low}\gets\emptyset$, $W\gets V$, $W_{low}\gets \{v\in W: |N(v)\cap (W\cup B_{high})|\leq \delta\alpha\}$}, $W_{high}\gets W\setminus W_{low}$
\State{Initialize counters: $\forall v\in V: c_v \leftarrow 0$}
\State{Initialize degrees: $\forall v\in V: |N(v)\cap W|=\deg(v)$, $|N(v)\cap (W\cup B_{high})|=\deg(v)$}

\While{$W\not=\emptyset$}
\State{Each vertex $v$ runs the following procedure:}

\If{$v\in W_{low}$}
\State $p(v)\gets$ a value in $[0, 1]$ chosen uniformly at random
\State {\bf Send} $p(v)$ message to neighbors
\EndIf
\If{$v\in B_{high}\cup W_{high}$}
\State {{\bf Send} \textsc{ack} to each $u\in N_{G_{bi}}(v)$ such that $p(u)$ is among the $\alpha$ smallest values $v$ received}
\EndIf
\If{$v\in W_{low}$ and $v$ receives \textsc{ack} from all $u\in N_{G_{bi}}(u)$}
        \State{Move $v$ to $D$}\label{line:moved}
        \State{{\bf Send} \textsc{increment counter} message to neighbors}\label{line:inc}
        \State{{\bf Send} \textsc{moved from $W$ to $D$} message to neighbors}
    \EndIf
    \State{Run \cref{alg:dist} starting from \cref{line:book}}
\EndWhile

\end{algorithmic}
\end{minipage}
\end{algorithm}

\subsection{Analysis}
We begin with a simple claim about the behavior of the partition of vertices over time:

\begin{claim}\label{claim:enter3}$ $
\begin{enumerate}
    \item No vertex can ever leave $D$.
    \item No vertex can ever enter $W\cup B_{high}$ from another set.
    \item No vertex in $W_{low}$ can ever at a later point be in $B_{high}\cup W_{high}$.
\end{enumerate}
\end{claim}
\begin{proof}
The proofs of items 1 and 2 follow from the proof of \cref{claim:enter}. Item 3 holds because if a vertex $v$ is in $B_{high}$ then $|N(v)\cap W|>\delta\alpha$ and if $v$ is in $W_{high}$ then $|N(v)\cap (W\cup B_{high})|>\delta\alpha$. For any $v$, the quantities $|N(v)\cap W|$ and $|N(v)\cap (W\cup B_{high})|$ can only decrease over time (they are only decremented in \cref{alg:distimp}).
\end{proof}

Next, we prove a simple claim that upper bounds the counter of each vertex:

\begin{claim}\label{claim:counter}
For all $v\in V$, at all times $c_v< 2\delta\alpha$.
\end{claim}
\begin{proof}
If for any $v\in V$, it happens that $c_v\geq\delta\alpha$, then during the same iteration of the {\bf while} loop, $v$ enters $D$ (on \cref{line:movecount}), after which point $v$ never leaves $D$ (by \cref{claim:enter3}) so $c_v$ cannot ever change again. Thus, it suffices to show that for all vertices $v\in V\setminus D$, during a single iteration of the {\bf while} loop, $c_v$ can be incremented at most $\delta\alpha$ times, leading to a maximum value of at most $2\delta\alpha-1$. This is true for $v\in W_{low}$ because $|N(v)\cap (W\cup B_{high})|\leq \delta\alpha$, and only vertices in $|N(v)\cap (W\cup B_{high})|$ can tell $v$ to increment $c_v$. On the other hand, if $v\in W\cup B_{high}$, then a neighbor $u$ of $v$ only sends \textsc{increment counter} if $p(u)$ is among the $\alpha$ smallest values in $v$'s neighborhood, so $c_v$ is only incremented $\alpha$ times during a single iteration. 
\end{proof}

The analysis of correctness is the same as that of our centralized algorithm (see \cref{sec:approx}), with one technicality: By \cref{claim:counter}, the counter of each vertex has maximum value $2\delta\alpha$ instead of $\delta\alpha$, causing an increase in the leading constant in the $O(\alpha)$ approximation factor. 

Our goal in the rest of this section is to prove that \cref{alg:distimp} runs in $O(\alpha\log n)$ rounds with high probability. Note that one iteration of the {\bf while} loop takes a constant number of rounds. Thus, our goal is to show that there are $O(\alpha\log n)$ total iterations of the loop. 


For any vertex $v\in V(G_{bi})$, let $N^2(v)$ be the set of vertices in the 2-hop neighborhood of $v$ with respect to $G_{bi}$, and let $E^2(v)$ be the set of edges within 2 hops of $v$ with respect to $G_{bi}$; that is, $E^2(v)$ contains the edge $(v,u)$ for all $u\in N_{G_{bi}}(v)$, and the edge $(u,y)$ for all $y\in N_{G_{bi}}(u)$.

We divide the iterations of the {\bf while} loop into two types. We say that an iteration is of \emph{type low degree} if at least half of the vertices $v\in W_{low}$ have $|N^2(v)|\leq \alpha$. Otherwise, we say that an iteration is of \emph{type high degree}. 

It is simple to deterministically bound the number of iterations of type low degree:

\begin{claim}
The total number of iterations of type low degree is $O(\log n)$. 
\end{claim}
\begin{proof}
By the specification of the algorithm, all vertices $v$ with $|N^2(v)|\leq \alpha$ are added to $D$ (on \cref{line:moved}). Thus, by definition, during an iteration of type low degree, at least half of the vertices in $W_{low}$ enter $D$. By \cref{claim:induced2}, at least half of the vertices in $W\cup B_{high}$ are in $W_{low}$. Thus, during an iteration of type low degree, at least $1/4$ of the vertices in $W\cup B_{high}$ enter $D$. By \cref{claim:enter3}, no vertex can ever enter $W\cup B_{high}$ from another set. Therefore, during every iteration of type low degree, $W\cup B_{high}$ shrinks by a factor of at least $4$. Thus, there are only $O(\log n)$ iterations of type low degree.
\end{proof}

It remains to bound the number of iterations of type high degree. Fix an iteration $I$ of type high degree.

For the purpose of analysis, we assign each edge $e\in E$ a \emph{weight} $w(e)$ that increases over the execution of the algorithm. The rule for updating the weight of edges is as follows. Whenever a vertex $v\in W_{low}$ is moved to $D$ on \cref{line:moved} (as a result of $p(v)$ being an $\alpha$-minimum), $v$ increments the weight of every edge in $E^2(v)$.

We will define the \emph{available weight} at iteration $I$ as a function that will capture the total amount of weight that could ever be added over all iterations starting from iteration $I$. Our goal is to provide:
\begin{enumerate}
    \item an upper bound of $2\delta\alpha^2\cdot|V(G_{bi})|$ for the available weight at iteration $I$, and
    \item a lower bound of $\alpha\cdot |V(G_{bi})|/4$ for the expected total weight added to edges during iteration $I$.
\end{enumerate}

Combining these upper and lower bounds yields the result that in each iteration the expected total amount of weight added is a $1/O(\alpha)$ fraction of the total available weight. This allows us to bound the number of iterations by $O(\alpha\log n)$ with high probability.

\begin{definition}
The \emph{available weight} at iteration $I$, denoted $A(I)$ is given by \[A(I)=\sum_{e\in E(G[W\cup B_{high}])} 2\delta\alpha-w(e)\] where the parameters in the expression are taken to be their values at the beginning of iteration $I$.
\end{definition}

Giving an upper bound on $A(I)$, which is item 1 of our above goal, is simple: 

\begin{claim}\label{claim:available}
$A(I)\leq 2\delta\alpha^2\cdot|V(G_{bi})|$.
\end{claim}
\begin{proof}
By definition all edge weights are non-negative, so $A(I)\leq \sum_{e\in E(G[W\cup B_{high}])} 2\delta\alpha$. Since $G[W\cup B_{high}]$ has arboricity at most $\alpha$, $|E(G[W\cup B_{high}])|\leq \alpha \cdot|W\cup B_{high}|=\alpha\cdot|V(G_{bi})|$. This completes the proof.
\end{proof}

Now, we consider item 2 of our above goal. 
Let $w(I)$ be a random variable denoting the aggregate total weight added to edges during iteration $I$. The randomness is over the choice of $p(v)$ for each vertex $v\in W_{low}$. Towards lower bounding $w(I)$, for every vertex $v\in W_{low}$ we define the random variable $R_v$ as the number of edges whose weight is incremented by $v$ during iteration $I$. That is, $w(I)=\sum_{v \in W_{low}} R_v$. We now calculate $\mathbb{E}[R_v]$.

\begin{claim}\label{claim:exp}
For all $v\in W_{low}$ with $|N^2(v)|\geq \alpha$, it holds that $\mathbb{E}[R_v]\geq \alpha$. 
\end{claim}

\begin{proof}
Note that $E^2(v)$ and $N^2(v)$ are taken to mean these value at the beginning of iteration $I$. A vertex $v\in W_{low}$ increments the weight of each edge in $E^2(v)$ if $v$ is an $\alpha$-minimum,
and otherwise $v$ does not increment the weight of any edges. 
A sufficient condition for $v$ to be an $\alpha$-minimum is that $p(v)$ is among the $\alpha$ smallest values in $N^2(v)$. Since $|N^2(v)|\geq \alpha$, the probability that $p(v)$ is among the $\alpha$ smallest values in $N^2(v)$ is $\alpha/|N^2(v)|\geq \alpha/|E^2(v)|$ since each vertex chooses its value uniformly at random. Therefore, $\mathbb{P}[R_v=|E^2(v)|]\geq \alpha/|E^2(v)|$. Thus, $\mathbb{E}[R_v]\geq \alpha$.
\end{proof}

We now give a lower bound on the expected aggregate weight $\mathbb{E}[w(I)]$:

\begin{claim}\label{claim:expi}
$\mathbb{E}[w(I)]\geq\alpha\cdot |V(G_{bi})|/4$
\end{claim}
\begin{proof}
Recall that $w(I)=\sum_{v \in W_{low}} R_v$. Thus,
\begin{align*}
\mathbb{E}[w(I)]&=\mathbb{E}\Big[\sum_{v\in W_{low}}R_v\Big]\\
&=\sum_{v\in W_{low}}\mathbb{E}[R_v]\\
&\geq\sum_{v\in W_{low},|N^2(v)|\geq \alpha}\mathbb{E}[R_v]\\
&\geq \sum_{v\in W_{low},|N^2(v)|\geq \alpha}\alpha&\text{(by \cref{claim:exp})}\\
&\geq \alpha\cdot|W_{low}|/2&\text{(since iteration $I$ is of type high degree)}\\
&\geq \alpha\cdot |B_{high}\cup W|/4 &\text{(by \cref{claim:induced2})}\\
&=\alpha\cdot |V(G_{bi})|/4. 
\end{align*}
\end{proof}

Before combining the above upper and lower bounds, we need to show that our function $A(I)$ accurately measures the progress of our algorithm by proving the following properties:

\begin{claim}$ $\label{claim:pot}
\begin{enumerate}
\item If $I'$ is the iteration right after $I$, then $A(I')\leq A(I)-w(I)$.
\item $A(I)>0$.

\end{enumerate}
\end{claim}
\begin{proof}
For item 1, it suffices to observe that the weight of an edge can only increase and the set of edges we sum over in the definition of $A(I)$ can only decrease. This is true because by \cref{claim:enter3}, no vertex can ever enter $W\cup B_{high}$ from another set.

For item 2, it suffices to show that $2\delta\alpha-w(e)$ is always positive.
Suppose for contradiction that there is an edge $(u,v)$ with $w(u,v)\geq2\delta\alpha$. Consider the point at which $w(u,v)$ was incremented to $\delta\alpha$. Consider $G_{bi}$ at this point in time. Only the weight of edges in $G_{bi}$ can be incremented, so $(u,v)\in E(G_{bi})$. Without loss of generality, $u\in B_{high}\cup W_{high}$ and $v\in W_{low}$. Ever since the edge $(u,v)$ entered $G_{bi}$, $u$ has been in $B_{high}\cup W_{high}$ and $v$ has been in $W_{low}$, since no vertex in $W_{low}$ can later be in $B_{high}\cup W_{high}$ by \cref{claim:enter3}. Each time $w(u,v)$ is incremented, it is caused by an \textsc{increment counter} message sent by some vertex $y$ that moved from $W_{low}$ to $D$, for which $(u,v)\in E^2(y)$. Then since $y$ was in $W_{low}$, $u\in B_{high}\cup W_{high}$, and $(u,v)\in E^2(y)$, we have $y\in N(u)$. Thus $y$ sends \textsc{increment counter} to $u$. Therefore, every time $w(u,v)$ is incremented, $c_u$ is also incremented. Since $w(u,v)=2\delta\alpha$, we have $c_u=2\delta\alpha$, which is a contradiction by \cref{claim:counter}. 
\end{proof}

We are now ready to put everything together to complete the analysis. For all $j$, let $I_j$ denote the $j^{th}$ iteration of type high degree. Then, by \cref{claim:available} we have that for all $j$, $A(I_j)\leq 2\delta\alpha^2\cdot|V(G_{bi})|$, and by \cref{claim:expi} we have that for all $j$, $\mathbb{E}[w(I_j)]\geq \alpha\cdot |V(G_{bi})|/4$ (where $G_{bi}$ is taken to be its value at the beginning of iteration $I_j$). Thus, $\mathbb{E}[w(I_j)]\geq A(I_j)/(8\delta\alpha)$.

Thus, by item 1 of \cref{claim:pot}, we have that if for all $j$, $A(I_{j+1})\leq A(I_j)-w(I_j)$, so for all $j$, we have
\begin{align*}
    \mathbb{E}[A(I_{j+1})]&\leq \mathbb{E}[A(I_j)-w(I_j)]\\
    &\leq (1-\frac{1}{8\delta\alpha})\mathbb{E}[A(I_j)]\\
    &\leq (1-\frac{1}{8\delta\alpha})^j\cdot A(I_1)\\
    &= (1-\frac{1}{8\delta\alpha})^j\cdot m(2\delta\alpha)\\
    &=(1-\frac{1}{32\alpha})^j\cdot 8\alpha^2n.
\end{align*}

We note that the above inequalities use the fact that the randomness in each iteration is independent. 

For any constant positive integer $c$, by setting $j=50c\alpha\log n$, we have $\mathbb{E}[A(I_{j+1})]<1/n^c$. By Markov's inequality and the fact that the available weight is always integral and non-negative, we have that $A(I_j)=0$ with high probability. By item 2 of \cref{claim:pot}, the algorithm terminates before it reaches an iteration $I^*$ with $A(I^*)=0$. Thus, the number of iterations of type high degree is $O(\alpha\log n)$ with high probability.


\subsection{Handling Unknown arboricity}\label{sec:unknown}

All of our distributed algorithms so far have assumed that $\alpha$ is known to each processor but that $n$ is unknown. In this section we will show that all of our distributed algorithms (\cref{thm:improved},  \cref{thm:local}, and \cref{thm:congest}) still work if the arboricity $\alpha$ is unknown to each processor, but $n$ is known. 

Let $\mathcal{A}$ be any one of our three distributed algorithms.  Let $R_{\mathcal{A}}^n(\alpha)$ be the upper bound on the number of rounds of algorithm $\mathcal{A}$ for a graph on $n$ vertices of arboricity $\alpha$ when alpha is known, where this upper bound is provided by \cref{thm:improved},  \cref{thm:local}, or \cref{thm:congest} depending on which algorithm $\mathcal{A}$ refers to. We will run $\mathcal{A}$ guessing at most $\log n$ different values for $\alpha$. For all $i$ from $0$ to $\log n$, let $\mathcal{A}_i$ be $\mathcal{A}$ run with $2^i$ as the guessed value of $\alpha$, where we stop after $R_{\mathcal{A}}^n(2^i)$ rounds. We note that each processor can compute the value of $R_{\mathcal{A}}^n(2^i)$ since $n$ is known.

The algorithm is as follows. We do the following until every vertex is dominated (i.e. until $W$ is empty). For each $i$ from $0,1,2\dots$ in order, we run $\mathcal{A}_i$. Importantly, we begin running $\mathcal{A}_i$ with the conditions present at the end of $\mathcal{A}_{i-1}$. That is, we begin $\mathcal{A}_i$ with the partition into $D$, $B$, and $W$ from the end of $\mathcal{A}_{i-1}$, and the values of all counters $c_v$ from the end of $\mathcal{A}_{i-1}$. 

We will now discuss the number of rounds. The goal is to show that the number of rounds is $O(R_{\mathcal{A}}^n(\alpha))$. If the algorithm terminates right after $\mathcal{A}_j$, then by definition the number of rounds is $\sum_{i=0}^{j}R_{\mathcal{A}}^n(2^i)$. Since $R_{\mathcal{A}}^n$ has at least linear dependence on its input for all of our algorithms $\mathcal{A}$, this is a geometric series with sum $O(R_{\mathcal{A}}^n(2^j))$. Thus, 
it suffices to show that the algorithm terminates (i.e. $W$ becomes empty) by the time $\mathcal{A}_{\lceil \log \alpha \rceil}$ terminates. This would be true by definition if $\mathcal{A}_{\lceil \log \alpha \rceil}$ were run from scratch on the original graph, without the initial conditions imposed by the previous $\mathcal{A}_i$s. 
When we do have these initial conditions however, we are instead running $\mathcal{A}_{\lceil \log \alpha \rceil}$ on an instance of the problem where some vertices have already been added to $D$. That is, $\mathcal{A}_{\lceil \log \alpha \rceil}$ is given an initial partition into $D\cup B \cup W$ and we only require $\mathcal{A}_{\lceil \log \alpha \rceil}$ to add vertices in $D$ to dominate the vertices in $W$ (i.e. the vertices that are not yet dominated). In other words, we are running $\mathcal{A}_{\lceil \log \alpha \rceil}$ on a graph where some progress is built into the initial conditions, so it is intuitive this algorithm satisfies the same running time guarantees as running $\mathcal{A}_{\lceil \log \alpha \rceil}$ from scratch, but this needs to be formally verified. Indeed, one can verify that all of our claims in the running time analyses of our three algorithms still hold given this initial condition. That is, the entire argument for each of our algorithms can be applied verbatim, after replacing $\alpha$ with $2^{\lceil \log \alpha \rceil}$.
As a result, our algorithm terminates by the time $\mathcal{A}_{\lceil \log \alpha \rceil}$ terminates.



We will now analyze the approximation factor. Let $OPT$ be a minimum dominating set. For all $i$, let $W_i$ be the value of the set $W$ right before executing $\mathcal{A}_i$. We define $OPT_i$ as the smallest set of vertices that dominates $W_i$. By definition, for all $i$ we have $OPT_i\leq OPT$. Let $D_i$ be the set of vertices added to $D$ during the execution of $\mathcal{A}_i$. Our goal is to show that for all $i$, $|D_i|=O(2^i\cdot|OPT_i|)$. The consequence of this is that since the algorithm terminates by the time $\mathcal{A}_{\lceil \log \alpha \rceil}$ terminates, we have that $D=O(\sum_{i=0}^{\log \alpha } 2^i\cdot|OPT_i|)=O(  \sum_{i=0}^{\log \alpha } 2^i\cdot|OPT|)=O(\alpha\cdot|OPT|)$. 

We will apply the analysis of correctness from our centralized algorithm (\cref{sec:approx}) separately for each $\mathcal{A}_i$. Fix $i$. The analysis of correctness of the centralized algorithm partitions $D$ into two sets: $D_{active}$ and $D_{passive}$. Let $D_{i_{active}} = D_i\cap D_{active}$ and let $D_{i_{passive}} = D_i\cap D_{passive}$. The proof of correctness for our centralized algorithm has two claims which together imply that $|D|=O(\alpha\cdot|OPT|)$:~\cref{claim:active} ($|D_{active}|\leq 2\delta\alpha\cdot|OPT|$), and~\cref{claim:passive} ($|D_{passive}|\leq|D_{active}|$). Both of these claims remain true for the analysis of $\mathcal{A}_i$ up to constant factors, after replacing $D_{active}$ and $D_{passive}$ with $D_{i_{active}}$ and $D_{i_{passive}}$ respectively, replacing $|OPT|$ with $|OPT_i|$, and replacing $\alpha$ with $2^i$. After these replacements, the proof of~\cref{claim:active} is identical to its original proof (except if the algorithm $\mathcal{A}$ is from \cref{thm:improved} then each counter $c_v$ is at most $2\delta\alpha$ instead of $\delta\alpha$). In particular, the original proof uses the fact that each vertex in $D_{active}$ is adjacent to a vertex in $OPT$, while the proof for $\mathcal{A}_i$ uses the fact that each vertex in $D_{i_{active}}$ is adjacent to a vertex in $OPT_i$. The proof of~\cref{claim:passive} for $\mathcal{A}_i$ has a slightly different constant factor than the original. In particular, in the original proof of~\cref{claim:passive}, we show that every vertex in $D_{passive}$ has at least $\delta\alpha$ neighbors in $D_{active}$, while every vertex in $D_{active}$ has at most $\delta\alpha$ neighbors in $D_{passive}$. The latter is still true in $\mathcal{A}_i$, but the bound for the former becomes $\delta\alpha/2$ since the counter of a vertex that reaches $\delta\alpha$ could have reached at most $\delta\alpha/2$ during $\mathcal{A}_1, \dots, \mathcal{A}_{i-1}$, and thus, must have incremented at least $\delta\alpha/2$ times during $\mathcal{A}_i$.






\paragraph{Acknowledgements} The authors would like to thank Yosi Hezi and Quanquan Liu for fruitful discussions. We would also like to thank Neal E. Young and Kent Quanrud for correspondence about their prior work.

\bibliographystyle{alpha}
\bibliography{ref}
\end{document}